\documentclass{fundam}

\usepackage{latexsym}
\usepackage{amssymb}
\usepackage{amsmath}
\usepackage[english]{babel}
\usepackage{url}

\DeclareMathOperator{\Hom}{Hom}
\DeclareMathOperator{\Tri}{Tri}

\usepackage{hyperref}

\begin{document}


\setcounter{page}{285}
\publyear{24}
\papernumber{2183}
\volume{191}
\issue{3-4}

\finalVersionForARXIV


\title{Commuting Upper Triangular Binary Morphisms}

\author{Juha Honkala\thanks{Address for correspondence: Department of Mathematics and Statistics,
                        University of Turku, FI-20014 Turku, Finland. }
\\
Department of Mathematics and Statistics\\
University of Turku\\
FI-20014 Turku, Finland\\
juha.honkala@utu.fi
}

\date{}

\maketitle

\runninghead{J. Honkala}{Commuting Upper Triangular Binary Morphisms}

\begin{abstract}
A morphism $g$ from the free monoid $X^*$ into itself is called upper triangular if the matrix of $g$ is upper triangular.
We characterize all upper triangular binary morphisms $g_1$ and $g_2$ such that $g_1g_2=g_2g_1$.
\end{abstract}


\section{Introduction}
The free monoid morphisms play an important role in many areas of mathematics and theoretical computer science (see \cite{AS,Lo1,Lo2,RS1,RS2,S}). On the other hand, many questions concerning combinatorics
on morphisms appear to be rather difficult. It is instructive to consider the problem of commutativity. If $u$ and $v$ are words, the equation $uv=vu$ holds if and only if there is a word $w$ such that $u$ and $v$ are powers of $w$ (see \cite{Lo1}). For free monoid morphisms the situation is more complicated. For two morphisms $g_1$ and $g_2$, the equation $g_1g_2=g_2g_1$ does not imply that $g_1$ and $g_2$ are powers of a third morphism (see, however, \cite{RW}).

In this paper we study commuting upper triangular binary morphisms. Let $X=\{a,b\}$ be a binary alphabet. A morphism $g$ from the free monoid $X^*$ into itself is called upper triangular if the matrix of $g$ is upper triangular. If $a$ is the first letter of $X$, this means that there is a nonnegative integer $s$ such that $g(a)=a^s$. We will characterize all upper triangular binary morphisms $g_1$ and $g_2$ such that $g_1g_2=g_2g_1$.

We now outline the contents of this paper. In Section 2 we recall the basic definitions. In Section 3 we discuss the connections between freeness and commutativity. In Section 4 we give examples of commuting morphisms. In Section 5 we study infinite words generated by morphisms. While the morphisms we study are not uniform, it turns out to be possible to use results concerning automatic sequences. In particular, we will apply the
theorem of Cobham characterizing those sequences which are automatic in two multiplicatively independent bases (see \cite{AS}).

In Sections 6,7,8 and 9 we characterize all upper triangular binary morphisms $g_1$ and $g_2$ such that $g_1g_2=g_2g_1$. Assume that $a$ is the first letter and $b$ is the second letter of the binary alphabet. In Section 6 we consider nonsingular morphisms such that both $g_1(b)$ and $g_2(b)$ contain at least two occurrences of $b$. We have two cases depending on whether these numbers are multiplicatively independent or not. The remaining cases are easier and are discussed in Sections 7,8 and 9.

We assume that the reader is familiar with the basics of free monoid morphisms, infinite words, automatic sequences and combinatorics on words (see \cite{AS,Lo1,Lo2,RS1,RS2,S}).
For previous results concerning combinatorics on morphisms see, e.g., \cite{Ho1,Ho2,Ho3,Ho4,RW}.

\section{Definitions}
We use standard notation and terminology concerning free monoids and their morphisms (see \cite{AS,Lo1,Lo2,RS1,RS2}).
If $X$ is a finite nonempty set, $X^*$ is the {\em free monoid} generated by $X$. The identity element of $X^*$ is the {\em empty word} denoted by $\varepsilon$.
If $u,v,w$ are words such that $uv=w$, we denote $v=u^{-1}w$.

If $w$ is a word and $a$ is a letter, then $|w|_a$ is the number of occurrences of $a$ in $w$.
The {\em length} of a word $w$, denoted by $|w|$, is the total number of letters in $w$.

\medskip
Let $X$ and $Y$ be finite nonempty alphabets.
A mapping $h:X^*\to Y^*$ is a {\em morphism} if
$$h(uv)=h(u)h(v)$$
for all $u,v\in X^*$.
The set of all morphisms from $X^*$ to $X^*$ is denoted by $\Hom(X^*)$. $\Hom(X^*)$ is a monoid with
respect to the usual product of morphisms.

\medskip
If $h\in \Hom(X^*)$ and the letters of $X$ are $x_1,\ldots,x_d$ in a fixed order, then the {\em  matrix} $M_h$ of $h$ is defined by
$$M_h=\left(\begin{array}{cccc}
                   |h(x_1)|_{x_1} & |h(x_2)|_{x_1}&\ldots &|h(x_d)|_{x_1}\\
                   |h(x_1)|_{x_2} & |h(x_2)|_{x_2}&\ldots &|h(x_d)|_{x_2}\\
                   \vdots&\vdots&&\vdots\\
                   |h(x_1)|_{x_d} & |h(x_2)|_{x_d}&\ldots &|h(x_d)|_{x_d}\\
                \end{array}  \right) .$$

A morphism $h\in \Hom(X^*)$ is {\em upper triangular} if its matrix $M_h$ is upper triangular. The set of upper triangular morphisms from $X^*$ to $X^*$ is denoted by $\Tri(X^*)$.
A morphism $h\in \Hom(X^*)$ is {\em nonsingular} if its matrix is nonsingular.

Let now $X$ be a finite alphabet and let
$h\in\Hom(X^*)$. If $w\in X^*$ is a word such that $w$ is a prefix of $h(w)$ and $\lim_{n\rightarrow \infty} |h^n(w)|=\infty$, we say that $h$ is {\em prolongable} on $w$ and define the infinite word $h^{\omega}(w)$ by
$$h^{\omega}(w)=\lim_{n\rightarrow  \infty}h^n(w).$$
Hence, $h^{\omega}(w)$ is the unique infinite word $u$ such that $h^n(w)$ is a prefix of $u$ for all $n\in{\mathbb N}$.

\section{Connections between freeness and commutativity}
A nonempty subset $Y$ of a semigroup $S$ is called  {\em free} if every element of the subsemigroup of $S$ generated by $Y$ can be written uniquely as a product of elements of $Y$. In other words, a set $Y$ is free if for all positive integers $m$ and $n$ and $u_1,\ldots,u_m,v_1,\ldots,v_n\in Y$, the equation
$$u_1u_2\cdots u_m=v_1v_2\cdots v_n$$
implies that
$$m=n \hspace{3mm} \mbox{and} \hspace{3mm} u_i=v_i \hspace{3mm} \mbox{for} \hspace{3mm} i=1,\ldots,m.$$

For an excellent introduction to freeness problems over semigroups we refer to \cite{CN}.

If a set contains two elements which commute, then the set is not free.
If $u,v\in X^*$ and $u\neq v$, then $\{u,v\}$ is free if and only if $u$ and $v$ do not commute (see \cite{Lo1}).

We recall some related results for upper triangular morphisms.

\medskip
First, let $X=\{a,b\}$ be a binary alphabet. Let $g_1,g_2\in \Tri(X^*)$.
We say that $\{g_1,g_2\}$ is a {\em special pair} if $g_1(b)$ and $g_2(b)$ belong to $a^*ba^*$ and exactly one of $g_1(a)$ and $g_2(a)$ equals $a$.

The following result is from \cite{Ho4}.

\begin{theorem}\label{the1}
Let $X=\{a,b\}$ and let $g_1,g_2\in \Tri(X^*)$ be nonsingular upper triangular morphisms. Assume that $g_1\neq g_2$. Assume that $\{g_1,g_2\}$ is not a special pair. If  $\{g_1,g_2\}$ is not free, then $g_1g_2=g_2g_1$.
\end{theorem}

For larger alphabets we have the following result (see \cite{Ho3}).

\begin{theorem}\label{the2}
Let $X$ be an arbitrary alphabet.
Let $g_1,g_2\in \Tri(X^*)$ and let $M_i$ be the matrix of $g_i$ for $i=1,2$. Assume $g_1\neq g_2$. Assume that all diagonal entries of
$M_i$ are at least two for $i=1,2$. If $\{g_1,g_2\}$ is not free, then $g_1g_2=g_2g_1$.
\end{theorem}

Theorems \ref{the1} and \ref{the2} imply the following lemma.

\begin{lemma}\label{comm}
Assume that the morphisms $g_1$ and $g_2$ satisfy the assumptions of Theorem \ref{the1} or Theorem \ref{the2}. Let $m$ and $n$ be positive integers. If $g_1^m$ and $g_2^n$ commute, then $g_1$ and $g_2$ commute.
\end{lemma}

 \begin{proof}
 Assume that $g_1^mg_2^n=g_2^ng_1^m$. Then the pair $\{g_1,g_2\}$ is not free and the claim follows by Theorem \ref{the1}
 or by Theorem \ref{the2}.
\end{proof}

\section{Examples}
In this section we give examples of commuting morphisms. The morphisms considered in Example~\ref{ex1} can be regarded as direct sums of unary morphisms.

\begin{example}\label{ex1}
{\rm  Let $X=\{x_1,\ldots,x_k\}$ be an alphabet having $k$ letters. Let
$(m_1,\ldots,m_k)$ and $(n_1,\ldots$, $n_k)$ be $k$-tuples of nonnegative integers. Define
the morphisms $g_1,g_2\in \Tri(X^*)$ by
$$g_1(x_i)=x_i^{m_i} \hspace{3mm} \mbox{ and } \hspace{3mm} g_2(x_i)=x_i^{n_i}$$
for $i=1,\ldots,k$. Then
$$g_1g_2(x_i)=g_1(x_i^{n_i})=x_i^{m_in_i}$$
and
$$g_2g_1(x_i)=g_2(x_i^{m_i})=x_i^{m_in_i}$$
for $i=1,\ldots,k$. Hence
$$g_1g_2=g_2g_1.$$}
\end{example}

\begin{example}\label{ex2}
{\rm
Let $X=\{a,b\}$ and define the morphisms $g_1,g_2\in \Tri(X^*)$ by
$$g_1(a)=a, \hspace{3mm} g_1(b)=b^2$$
and
$$g_2(a)=a^2, \hspace{3mm} g_2(b)=b.$$
By Example \ref{ex1} the morphisms $g_1$ and $g_2$ commute. However, there do not exist positive integers $m$, $n$ and a morphism $g\in \Hom(X^*)$ such that
$g_1=g^m$ and $g_2=g^n$. To see this, assume that such $g$, $m$ and $n$ exist. Then neither $g(a)$ nor $g(b)$ is the empty word. Furthermore, either $|g(a)|=1$
or $|g(b)|=1$ but not both. Without loss of generality assume that $|g(a)|=1$. Then $g(a)=a$ or $g(a)=b$. The first alternative is not possible since $g^n(a)=a^2$.
The second alternative is not possible since it would imply that the only word of length one in $g(X^*)$ is $b$.}
\end{example}

\begin{example}\label{ex3}
{\rm
Let $X=\{a,b\}$ and let $p$ and $q$ be positive integers. Let $\alpha$ be a nonnegative integer. Define the morphisms $g_1,g_2\in \Tri(X^*)$ by
$g_1(a)=g_2(a)=a$, $g_1(b)=(ba^{\alpha})^{p-1}b$, $g_2(b)=(ba^{\alpha})^{q-1}b$.

To prove the equation $g_1g_2=g_2g_1$, let $z=ba^{\alpha}$. Then $g_1(z)=z^p$ and $g_2(z)=z^q$. Hence $g_1g_2(z)=g_2g_1(z)$. Therefore $g_1g_2(b)a^{\alpha}=g_2g_1(b)a^{\alpha}$, which implies that $g_1g_2(b)=g_2g_1(b)$. Trivially $g_1g_2(a)=g_2g_1(a)$.}
\end{example}

\begin{example}\label{ex4}
{\rm Let $X=\{a,b\}$ and let $g_1,g_2\in \Tri(X^*)$ be nonsingular upper triangular binary morphisms. Assume that there exist positive integers $m$ and $n$ such that
$g_1^m=g_2^n$. Assume $g_1\neq g_2$. Then $\{g_1,g_2\}$ is not a special pair. Indeed, if one of $g_1(a)$ and $g_2(a)$ equals $a$, then both do. Hence Theorem~\ref{the1} implies that $g_1g_2=g_2g_1$.}
\end{example}

Let $u$ and $v$ be words over the binary alphabet $X=\{a,b\}$. We say that $u$ and $v$ are {\em $a$-conjugates} if there exist nonnegative integers $p,q,r,s$ and a word $w$ such that
$$u=a^pwa^q, \hspace{3mm} v=a^rwa^s \hspace{3mm} \mbox{and} \hspace{3mm} p+q=r+s.$$
\eject

\begin{example}\label{ex5}
{\rm Let $X=\{a,b\}$ and let $g_1,g_2\in \Tri(X^*)$ be nonsingular upper triangular morphisms. Assume that $g_1(a)=g_2(a)=a$. Assume that there are positive integers $m$ and $n$ such that $g_1^n(b)$ and $g_2^m(b)$ are $a$-conjugates. We show that these conditions imply that $g_1$ and $g_2$ commute. By Lemma \ref{comm} it is enough to show that $g_1^n$ and $g_2^m$ commute.

\medskip
By assumption, there exist nonnegative integers $\gamma_1,\gamma_2,\delta_1,\delta_2,\alpha_1,\ldots,\alpha_{p-1}$ such that
$g_1^n(b)=a^{\gamma_1}za^{\gamma_2}$ and $g_2^m(b)=a^{\delta_1}za^{\delta_2}$, where
$z=ba^{\alpha_1}ba^{\alpha_2}b\cdots ba^{\alpha_{p-1}}b$ and $\gamma_1+\gamma_2=\delta_1+\delta_2$. Then
\begin{eqnarray*}
g_1^ng_2^m(b)&=&a^{\delta_1}g_1^n(z)a^{\delta_2}\\
&=&a^{\delta_1+\gamma_1}za^{\gamma_2+\alpha_1+\gamma_1}za^{\gamma_2+\alpha_2+\gamma_1}za^{\gamma_2}\cdots a^{\gamma_1}za^{\gamma_2+\alpha_{p-1}+\gamma_1}za^{\gamma_2+\delta_2}\vspace*{-3mm}
\end{eqnarray*}
and\vspace*{-3mm}
\begin{eqnarray*}
g_2^mg_1^n(b)&=&a^{\gamma_1}g_2^m(z)a^{\gamma_2}\\
&=&a^{\gamma_1+\delta_1}za^{\delta_2+\alpha_1+\delta_1}za^{\delta_2+\alpha_2+\delta_1}za^{\delta_2}\cdots a^{\delta_1}za^{\delta_2+\alpha_{p-1}+\delta_1}za^{\delta_2+\gamma_2}.
\end{eqnarray*}
Therefore $g_1^ng_2^m(b)=g_2^mg_1^n(b)$. Hence $g_1^ng_2^m=g_2^mg_1^n$.}
\end{example}

We conclude this section by two examples involving singular morphisms.

\begin{example}\label{ex6}
{\rm Let $X=\{a,b\}$. Define the morphisms $g_1,g_2\in \Tri(X^*)$ by
$$g_1(a)=g_2(a)=\varepsilon, \hspace{3mm} g_1(b)=w^i, \hspace{3mm} g_2(b)=w^j$$
where $w\in X^*$ and $i$ and $j$ are nonnegative integers. Then
$$g_1g_2(b)=g_1(w^j)=w^{ij|w|_b} \hspace{3mm} \mbox{and} \hspace{3mm} g_2g_1(b)=g_2(w^i)=w^{ij|w|_b}.$$
Hence $g_1g_2=g_2g_1$.}
\end{example}

\begin{example}\label{ex7}
{\rm  Let $X=\{a,b\}$. Define the morphisms $g_1,g_2\in \Tri(X^*)$ by
$$g_1(a)=\varepsilon, \hspace{3mm} g_1(b)=(a^{\alpha}ba^{\beta})^i$$
and
$$g_2(a)=a, \hspace{3mm} g_2(b)=(ba^{\alpha+\beta})^jb$$
where $\alpha,\beta,i,j$ are nonnegative integers. Then
$$g_1g_2(b)=g_1((ba^{\alpha+\beta})^jb)=g_1(b^{j+1})=(a^{\alpha}ba^{\beta})^{i(j+1)}$$
and
$$g_2g_1(b)=g_2((a^{\alpha}ba^{\beta})^i)=(a^{\alpha}(ba^{\alpha+\beta})^jba^{\beta})^i=(a^{\alpha}ba^{\beta})^{i(j+1)}.$$
Hence $g_1g_2=g_2g_1$.}
\end{example}

\section{Properties of infinite words generated by upper triangular binary morphisms}
Let $X=\{a,b\}$ be a binary alphabet. Regard $a$ as the first letter of $X$ and $b$ as the second letter of $X$.

\eject

Let $h\in \Tri(X^*)$. Assume that $h$ is nonsingular. Then there exist a nonnegative integer $\gamma$ and a word $v$ such that $h(b)=a^{\gamma}bv$.
Let $c$ be a new letter and let $Y=X\cup \{c\}$. Regard $c$ as the third letter of $Y$. Define the morphism ${\bf RIGHT}(h) \in \Tri(Y^*)$ by
$$\mbox{{\bf RIGHT}}(h)(x)=h(x), \hspace{3mm} \mbox{ if } x\in X, \hspace{6mm} \mbox{{\bf RIGHT}}(h)(c)=cv.$$
Assume that $v\neq \varepsilon$. Then we define the infinite word $\omega(h)$ by
$$\omega(h)=bc^{-1}{\bf RIGHT}(h)^{\omega}(c).$$
In other words, the infinite word $\omega(h)$ is obtained from ${\bf RIGHT}(h)^{\omega}(c)$ by replacing its first letter $c$ by $b$.
Hence, if $n$ is any positive integer, the word obtained from $h^n(b)$ by deleting all occurrences of $a$ preceding the first occurrence of $b$
is a prefix of $\omega(h)$.

For the proof of the following lemma see \cite{Ho3}.

\begin{lemma}\label{ba}
Let $g_1,g_2\in \Tri(X^*)$ be nonsingular morphisms. Let $h_i={\bf RIGHT}(g_i)$ for $i=1,2$. Assume that $h_i(c)\neq c$ for $i=1,2$. If $g_1g_2=g_2g_1$, then $\omega(g_1)=\omega(g_2)$.
\end{lemma}

We will now study some properties of the infinite words defined above.

Let $w$ be an infinite word over $X$ having infinitely many occurrences of $b$. For $i\geq 1$, let $A_w(i)$ be the number of occurrences of the letter $a$ in $w$ between the $i$th and the $(i+1)$th occurrences of $b$ in $w$.

The following lemma gives a formula for $A_w(i)$, which will be used repeatedly.

\begin{lemma}\label{le1}
Let $h\in \Tri(X^*)$  be  the morphism defined by
$$h(a)=a^s \hspace{3mm} \mbox{ and } \hspace{3mm} h(b)=a^{\gamma_1}ba^{\alpha_1}ba^{\alpha_2}b\cdots ba^{\alpha_{p-1}}ba^{\gamma_2},$$
where $s\geq1$, $p \geq 2$ and
$\gamma_1,\gamma_2,\alpha_1,\ldots,\alpha_{p-1}\geq 0$. Let $w=\omega(h)$. Then \smallskip

(i) $A_w(i+pn)=\alpha_i$ if $i\in \{1,\ldots,p-1\}$ and $n\geq 0$.

(ii) $A_w(pi)=sA_w(i)+\gamma_1+\gamma_2$ if $i\geq 1$.

(iii) If $m\geq 1$, $k\geq 0$, $d_m,\ldots,d_{m+k}\in \{0,1,\ldots,p-1\}$ and $d_m\neq 0$, then
$$
A_w(d_mp^m+d_{m+1}p^{m+1}+\cdots +d_{m+k}p^{m+k})=\alpha_{d_m}s^m+(\gamma_1+\gamma_2)(1+s+\cdots + s^{m-1}).
$$
\end{lemma}

\begin{proof}
The infinite word $w$ belongs to $a^{-\gamma_1}h(b)\{a,h(b)\}^{\omega}$ and $|h(b)|_b=p$. This implies (i).

To prove (ii), let
$$w=w_1ba^jb\cdots,$$
where $|w_1b|_b=i$ and $j=A_w(i)$. Then
$$w=a^{-\gamma_1}\, h(w_1)\, h(b)\, a^{js}\, h(b)\,  \cdots,$$
where $|h(w_1)h(b)|_b=p|w_1b|_b=pi$. Hence
$$A_w(pi)=\gamma_2+js+\gamma_1=sA_w(i)+\gamma_1+\gamma_2.$$
This proves (ii).

\eject

If $m=1$, (iii) is a consequence of (i) and (ii). Assume inductively that (iii) holds for $m\geq 1$. Assume that $k\geq 0$, $e_{m+1},\ldots,e_{m+1+k}\in \{0,1,\ldots,p-1\}$ and $e_{m+1}\neq 0$. Then
$$A_w(e_{m+1}p^{m+1}+e_{m+2}p^{m+2}+\cdots + e_{m+k+1}p^{m+k+1})$$
$$=sA_w(e_{m+1}p^m+e_{m+2}p^{m+1}+\cdots+e_{m+k+1}p^{m+k})+\gamma_1+\gamma_2$$
$$=s(\alpha_{e_{m+1}}s^m+(\gamma_1+\gamma_2)(1+s+\cdots+s^{m-1}))+\gamma_1+\gamma_2$$
$$=\alpha_{e_{m+1}}s^{m+1}+(\gamma_1+\gamma_2)(1+s+\cdots +s^m).$$
Here the first equation follows by (ii) and the second equation by the inductive hypothesis. This proves~(iii).
\end{proof}

The final lemma of this section studies the case of eventually periodic words.

\begin{lemma}\label{le2}
Let $h$ be as in Lemma \ref{le1}. Assume that $w=\omega(h)$ is eventually periodic. Then $\gamma_1=\gamma_2=0$ and $\alpha_1=\alpha_2=\cdots =\alpha_{p-1}$.
\end{lemma}

 \begin{proof} Since $w$ is eventually periodic, also the sequence $(A_w(i))_{i\geq 1}$ is eventually periodic. In particular, this sequence takes only finitely many different values.
Hence, by Lemma \ref{le1}, we have $\gamma_1=\gamma_2=0$. If $\alpha_1=\cdots=\alpha_{p-1}=0$, the claim of the lemma holds. Assume that some $\alpha_i$ is nonzero. Then Lemma \ref{le1} implies that $s=1$.

\medskip
Assume
$$A_w(i)=A_w(i+d) \hspace{3mm} \mbox{ for all } i\geq i_0,$$
where $i_0$ is an integer and $d=ep^m+fp^{m+1}$ for some $m\geq 0$, $e\in \{1,\ldots,p-1\}$ and $f\geq 0$. Choose an integer $n$ such that $n>m$ and $p^n \geq i_0$. Then
$$A_w(jp^n)=A_w(ep^m+fp^{m+1}+jp^n)$$
for $j=1,\ldots,p-1$.
Now Lemma \ref{le1} implies that
$$\alpha_j=\alpha_e$$
for $j=1,\ldots,p-1$. This implies the claim.
\end{proof}

\section{Commuting nonsingular morphisms $g_1$ and $g_2$ such that both $g_1(b)$ and $g_2(b)$ have at least two occurrences of $b$}
In this section $X=\{a,b\}$. We will consider nonsingular morphisms $g_1,g_2\in \Tri(X^*)$ such that $|g_i(b)|_b\geq 2$ for $i=1,2$. We have two different cases to consider according to whether $|g_1(b)|_b$ and $|g_2(b)|_b$ are multiplicatively independent or not.
Recall that two integers $p\geq 2$ and $q\geq 2$ are {\em multiplicatively dependent} if there are positive integers $r,m,n$ such that $p=r^m$ and $q=r^n$ (see \cite{AS}).

\subsection{The numbers of occurrences of $b$ are multiplicatively independent}

\begin{lemma}\label{le3}
Let $g_i\in \Tri(X^*)$, $i=1,2$, be morphisms such that
$$g_1(a)=a^s \hspace{3mm} \mbox{ and } \hspace{3mm} g_1(b)=a^{\gamma_1}ba^{\alpha_1}ba^{\alpha_2}b\cdots ba^{\alpha_{p-1}}ba^{\gamma_2}$$
and
$$g_2(a)=a^t \hspace{3mm} \mbox{ and } \hspace{3mm} g_2(b)=a^{\delta_1}ba^{\beta_1}ba^{\beta_2}b\cdots ba^{\beta_{q-1}}ba^{\delta_2}$$
where $s,t\geq1$, $p,q \geq 2$ and
$\gamma_1,\gamma_2,\delta_1,\delta_2,\alpha_1,\ldots,\alpha_{p-1},\beta_1,\ldots,\beta_{q-1}\geq 0$.
Assume that $p$ and $q$ are multiplicatively independent.  Assume that $g_1(b)\not\in b^*$.
Assume that $\omega(g_1)=\omega(g_2)$. Then $s=t=1$ and $\gamma_1=\gamma_2=\delta_1=\delta_2=0$.
\end{lemma}

 \begin{proof}
 Let $z$ be the smallest positive integer such that $\beta_z=\max\{\beta_i\mid i=1,2,\ldots,q-1\}$. Then $\beta_z\geq 0$ but it is possible that $\beta_z=0$.

\medskip
By Lemma \ref{le1} we have
\begin{equation}\label{equ2}
A_{\omega(g_2)}(i)\leq A_{\omega(g_2)}(zq^n)
\end{equation}
for $n\geq 1$ and $i<zq^n$.
Consider the numbers $zq^n$, $n\geq 1$. For $n\geq 1$, let
$$zq^n=p^{\tau(n)}(i_n+pj_n),$$
where $\tau(n),j_n\geq 0$ and $i_n\in \{1,\ldots,p-1\}$.

\medskip
Now, the set $\{j_n\mid n\geq 1\}$ is infinite. To see this, assume on the contrary that it is finite. Then there are integers $m$ and $n$ such that $i_m+pj_m=i_n+pj_n$ and $m<n$. This
implies that
$$\frac{zq^m}{p^{\tau(m)}}=\frac{zq^n}{p^{\tau(n)}}\ .$$
Hence $p^{\tau(n)-\tau(m)}=q^{n-m}$, which contradicts the assumption. It follows that the set $\{j_n\mid n\geq 1\}$ is infinite.
Therefore there is an integer $n\geq 1$ such that
$$zq^n=p^{\tau(n)}(i_n+x_1p+\cdots +x_kp^k)$$
where $k\geq 2$ and $x_k\neq 0$.

\medskip
Next, let $y$ be an integer such that $\alpha_y=\max \{\alpha_i\mid i=1,\ldots,p-1\}$ and consider the numbers
$$K_1=yp^{\tau(n)+k-1} \hspace{3mm} \mbox{ and } \hspace{3mm} K_2=zq^n=p^{\tau(n)}(i_n+x_1p+\cdots+x_kp^k).$$
Then we have $K_1<K_2$. Therefore (\ref{equ2}) implies that
$$A_{\omega(g_1)}(K_1)=A_{\omega(g_2)}(K_1)\leq A_{\omega(g_2)}(K_2)=A_{\omega(g_1)}(K_2).$$
On the other hand, Lemma \ref{le1} implies that
$$A_{\omega(g_1)}(K_1)=\alpha_y s^{\tau(n)+k-1} +(\gamma_1+\gamma_2)(1+s+\cdots + s^{\tau(n)+k-2})$$
and
$$A_{\omega(g_1)}(K_2)=\alpha_{i_n} s^{\tau(n)} +(\gamma_1+\gamma_2)(1+s+\cdots + s^{\tau(n)-1}).$$
Since $A_{\omega(g_1)}(K_1)\leq A_{\omega(g_1)}(K_2)$, we have $\gamma_1=\gamma_2=0$. If ${\alpha}_y=0$, we would have $g_1(b)\in b^*$ which contradicts our assumption. Hence $\alpha_y\neq 0$ and $s=1$.

\eject
Since $\gamma_1=\gamma_2=0$ and $s=1$, we have $A_{\omega(g_1)}(i)\in \{\alpha_1,\ldots,\alpha_{p-1}\}$ for all $i\geq 1$. Now the equality $\omega(g_1)=\omega(g_2)$ implies that $\delta_1=\delta_2=0$ and $t=1$.
\end{proof}

The next theorem gives all nonsingular morphisms $g_i\in \Tri(X^*)$, $i=1,2$, such that $g_1g_2=g_2g_1$ and the numbers $|g_1(b)|_b$ and $|g_2(b)|_b$ are multiplicatively independent integers larger than one. In the proof we use automatic sequences and Cobham's theorem characterizing sequences which are $p$-automatic and $q$-automatic for multiplicatively independent integers $p$ and $q$ (see \cite{AS}).

\begin{theorem}\label{the3}
Let $g_i\in \Tri(X^*)$, $i=1,2$,  be morphisms such that
$$g_1(a)=a^s \hspace{3mm} \mbox{ and } \hspace{3mm} g_1(b)=a^{\gamma_1}ba^{\alpha_1}ba^{\alpha_2}b\cdots ba^{\alpha_{p-1}}ba^{\gamma_2}$$
and
$$g_2(a)=a^t \hspace{3mm} \mbox{ and } \hspace{3mm} g_2(b)=a^{\delta_1}ba^{\beta_1}ba^{\beta_2}b\cdots ba^{\beta_{q-1}}ba^{\delta_2}$$
where $s,t\geq1$, $p,q \geq 2$ and
$\gamma_1,\gamma_2,\delta_1,\delta_2,\alpha_1,\ldots,\alpha_{p-1},\beta_1,\ldots,\beta_{q-1}\geq 0$.
Assume that $p$ and $q$ are multiplicatively independent.
Then $g_1g_2=g_2g_1$ if and only if at least one of the following conditions holds:\smallskip

(i) $g_i(b)\in b^*$ for $i=1,2$,

(ii) $g_1(a)=g_2(a)=a$, $g_1(b)=(ba^{\alpha})^{p-1}b$ and $g_2(b)=(ba^{\alpha})^{q-1}b$, where $\alpha=\alpha_1$.
\end{theorem}

\begin{proof}
If (i) or (ii) holds, then $g_1g_2=g_2g_1$ (see Examples \ref{ex1} and \ref{ex3}).

Assume that $g_1g_2=g_2g_1$. By Lemma \ref{ba} we have $\omega(g_1)=\omega(g_2)$. Hence, if $g_1(b)\in b^*$, also $g_2(b)\in b^*$ and (i) holds.
Assume that $g_1(b)\not\in b^*$ and $g_2(b)\not\in b^*$.

Now Lemma \ref{le3} implies that $s=t=1$ and $\gamma_1=\gamma_2=\delta_1=\delta_2=0$. By Lemma \ref{le1}, the sequence $(A_{\omega(g_1)}(i))_{i\geq 1}$
is $p$-automatic and the sequence $(A_{\omega(g_2)}(i))_{i\geq 1}$ is $q$-automatic. Since these sequences are equal and the numbers $p$ and $q$ are multiplicatively independent,
$(A_{\omega(g_1)}(i))_{i\geq 1}$ is eventually periodic. Hence $\omega(g_1)$ is eventually periodic. Now Lemma \ref{le2} implies that $\alpha_1=\alpha_2=\cdots =\alpha_{p-1}$.
Hence $g_1(b)=(ba^{\alpha})^{p-1}b$ where $\alpha=\alpha_1$.

A similar argument shows that $g_2(b)=(ba^{\beta})^{q-1}b$, where $\beta=\beta_1$. Since $\omega(g_1)=\omega(g_2)$, we have $\alpha_1=\beta_1$. Hence (ii) holds.
\end{proof}

\subsection{The numbers of occurrences of $b$ are multiplicatively dependent}
In this subsection we first consider the case that $|g_1(b)|_b$ and $|g_2(b)|_b$ are equal.

\begin{lemma}\label{le8}
Let $g_i\in \Tri(X^*)$, $i=1,2$, be morphisms such that
$$g_1(a)=a^s \hspace{3mm} \mbox{ and } \hspace{3mm} g_1(b)=a^{\gamma_1}ba^{\alpha_1}ba^{\alpha_2}b\cdots ba^{\alpha_{p-1}}ba^{\gamma_2}\vspace*{-1mm} $$
and \vspace*{-1mm}
$$g_2(a)=a^t \hspace{3mm} \mbox{ and } \hspace{3mm} g_2(b)=a^{\delta_1}ba^{\beta_1}ba^{\beta_2}b\cdots ba^{\beta_{q-1}}ba^{\delta_2}$$
where $s,t\geq1$, $p,q \geq 2$ and
$\gamma_1,\gamma_2,\delta_1,\delta_2,\alpha_1,\ldots,\alpha_{p-1},\beta_1,\ldots,\beta_{q-1}\geq 0$.
Assume that $p=q$. If $g_1g_2=g_2g_1$, then at least one of the following conditions holds:\smallskip

(i) $g_1=g_2$,

(ii) $g_i(b)\in b^*$ for $i=1,2$,

(iii) $g_1(a)=g_2(a)=a$ and the words $g_1(b)$ and $g_2(b)$ are $a$-conjugates.
\end{lemma}

  \begin{proof}
  Assume that $g_1g_2=g_2g_1$. Let $w_i=\omega(g_i)$ for $i=1,2$. By Lemma \ref{ba} we have $w_1=w_2$. Since $p=q$, the equation $w_1=w_2$ implies that $\alpha_i=\beta_i$ for $i=1,\ldots,p-1$.

If $g_1(b)\in b^*$, we have $w_1=b^{\omega}$. This implies that $g_2(b)\in b^*$ and hence (ii) holds.

Assume that $g_1(b)\not\in b^*$ and $g_2(b)\not \in b^*$.

Next, assume that $A_{w_1}(i)$ takes only finitely many different values. Then Lemma \ref{le1} implies that $\gamma_1=\gamma_2=0$. Since $g_1(b)\not\in b^*$, some $\alpha_i$ is nonzero. This implies that $s=1$. By a similar reasoning it is seen that $\delta_1=\delta_2=0$ and $t=1$. Then $g_1(a)=g_2(a)$ and $g_1(b)=g_2(b)$ and condition (i)
holds.

\medskip
Assume then that $A_{w_1}(i)$ takes infinitely many values.
Since $A_{w_1}(pi)=A_{w_2}(pi)$ for all $i\geq 1$, Lemma \ref{le1} implies that
$$sA_{w_1}(i)+\gamma_1+\gamma_2=tA_{w_2}(i)+\delta_1+\delta_2$$
for all $i\geq 1$. Because this equation holds for infinitely many different values of $A_{w_1}(i)=A_{w_2}(i)$, it follows that $s=t$ and $\gamma_1+\gamma_2=\delta_1+\delta_2$.
If now $s=t=1$, we have (iii).

\medskip
Assume that $s=t>1$. By counting the number of occurrences of $a$ before the first occurrence of $b$ in $g_1g_2(b)=g_2g_1(b)$, we see that
$$s\delta_1+\gamma_1=t\gamma_1+\delta_1.$$
By counting the number of occurrences of $a$ after the last occurrence of $b$ in $g_1g_2(b)=g_2g_1(b)$, we see that
$$\gamma_2+s\delta_2=\delta_2 +t\gamma_2.$$
Hence $(s-1)\delta_1=(t-1)\gamma_1$ and $(s-1)\delta_2=(t-1)\gamma_2$. Since $s=t>1$ we have $\gamma_1=\delta_1$ and $\gamma_2=\delta_2$. Therefore condition (i) holds.
\end{proof}

The next theorem gives all nonsingular morphisms $g_i\in \Tri(X^*)$, $i=1,2$, such that $g_1g_2=g_2g_1$ and the numbers $|g_1(b)|_b$ and $|g_2(b)|_b$ are
multiplicatively dependent integers larger than one.

\begin{theorem}\label{the4}
Let $g_i\in \Tri(X^*)$, $i=1,2$, be morphisms such that
$$g_1(a)=a^s \hspace{3mm} \mbox{ and } \hspace{3mm} g_1(b)=a^{\gamma_1}ba^{\alpha_1}ba^{\alpha_2}b\cdots ba^{\alpha_{p-1}}ba^{\gamma_2} \vspace*{-1mm}$$
and\vspace*{-1mm}
$$g_2(a)=a^t \hspace{3mm} \mbox{ and } \hspace{3mm} g_2(b)=a^{\delta_1}ba^{\beta_1}ba^{\beta_2}b\cdots ba^{\beta_{q-1}}ba^{\delta_2}$$
where $s,t\geq1$, $p,q \geq 2$ and
$\gamma_1,\gamma_2,\delta_1,\delta_2,\alpha_1,\ldots,\alpha_{p-1},\beta_1,\ldots,\beta_{q-1}\geq 0$.
Assume that $p=r^m$ and $q=r^n$ where $m,n,r$ are positive integers. Then $g_1g_2=g_2g_1$ if and only if at least one of the following
conditions holds:\smallskip

(i) $g_1^n=g_2^m$,

(ii) $g_i(b)\in b^*$ for $i=1,2$,

(iii) $g_1(a)=g_2(a)=a$ and the words $g_1^n(b)$ and  $g_2^m(b)$ are $a$-conjugates.
\end{theorem}

\begin{proof}
If at least one of the conditions (i), (ii) or (iii) holds, then $g_1g_2=g_2g_1$ (see Examples \ref{ex1}, \ref{ex4} and \ref{ex5}).

Assume $g_1g_2=g_2g_1$. Let $h_1=g_1^n$ and $h_2=g_2^m$. Then $|h_1(a)|_a=s^n$, $|h_2(a)|_a=t^m$ and $|h_1(b)|_b=p^n=r^{mn}=q^m=|h_2(b)|_b$.

Since $h_1h_2=h_2h_1$,  Lemma \ref{le8} implies that at least one of the following conditions holds:\smallskip

(a) $h_1=h_2$,

(b) $h_i(b)\in b^*$ for $i=1,2$,

(c) $h_1(a)=h_2(a)=a$ and the words $h_1(b)$ and $h_2(b)$ are $a$-conjugates.

Now (a) implies (i), (b) implies (ii) and (c) implies (iii).
\end{proof}

\section{Commuting nonsingular morphisms $g_1$ and $g_2$ such that $|g_1(b)|_b=1$ and $|g_2(b)|_b\geq 2$}
Let $X=\{a,b\}$. The following theorem gives all commuting nonsingular morphisms $g_i\in \Tri(X^*)$, $i=1,2$, such that $|g_1(b)|_b=1$ and $|g_2(b)|_b\geq 2$.

\begin{theorem}
Let $g_i \in \Tri(X^*)$, $i=1,2$,  be morphisms such that
$$g_1(a)=a^s \hspace{3mm} \mbox{ and } \hspace{3mm} g_1(b)=a^{\gamma_1}ba^{\gamma_2}$$
and
$$g_2(a)=a^t \hspace{3mm} \mbox{ and } \hspace{3mm} g_2(b)=a^{\delta_1}ba^{\beta_1}ba^{\beta_2}b\cdots ba^{\beta_{q-1}}ba^{\delta_2}$$
where $s,t\geq1$, $q \geq 2$ and
$\gamma_1,\gamma_2,\delta_1,\delta_2,\beta_1,\ldots,\beta_{q-1}\geq 0$.
Then $g_1g_2=g_2g_1$ if and only if at least one of the following conditions holds:\smallskip

(i) $g_1(x)=x$ for all $x\in \{a,b\}$,

(ii) $g_i(b)\in b^*$ for $i=1,2$.
\end{theorem}

\begin{proof}
 If (i) or (ii) holds, then $g_1g_2=g_2g_1$.

\medskip
Assume then that $g_1g_2=g_2g_1$. Let $h_1=g_1g_2$ and $h_2=g_2$. Then $h_1h_2=h_2h_1$. Since $|h_1(b)|_b=|h_2(b)|_b\geq 2$,
Lemma \ref{le8} implies that at least one of the following conditions holds:\smallskip

(a) $h_1=h_2$,

(b) $h_i(b)\in b^*$ for $i=1,2$,

(c) $h_1(a)=h_2(a)=a$ and the words $h_1(b)$ and $h_2(b)$ are $a$-conjugates.\medskip

Now (a) implies (i), (b) implies (ii) and (c) implies (i).
\end{proof}

\section{Commuting nonsingular morphisms $g_1$ and $g_2$ such that $|g_1(b)|_b=|g_2(b)|_b=1$}
Let $X=\{a,b\}$.
In this section we give all commuting nonsingular morphisms $g_i\in \Tri(X^*)$, $i=1,2$, such that $|g_1(b)|_b=|g_2(b)|_b=1$.

\begin{proposition}
Let $g_i\in \Tri(X^*)$, $i=1,2$, be morphisms such that
$$g_1(a)=a^s \hspace{3mm} \mbox{ and } \hspace{3mm} g_1(b)=a^{\gamma_1}ba^{\gamma_2}$$
and
$$g_2(a)=a^t \hspace{3mm} \mbox{ and } \hspace{3mm} g_2(b)=a^{\delta_1}ba^{\delta_2}$$
where $s,t\geq1$ and $\gamma_1,\gamma_2,\delta_1,\delta_2\geq 0$.
Then $g_1g_2=g_2g_1$ if and only if
$$(s-1)\delta_i=(t-1)\gamma_i \hspace{3mm} \mbox{ for } \hspace{3mm} i=1,2.$$
\end{proposition}

\begin{proof}
 We have
$$g_1g_2(b)=a^{s\delta_1+\gamma_1}ba^{s\delta_2+\gamma_2} \hspace{3mm} \mbox{and} \hspace{3mm} g_2g_1(b)=a^{t\gamma_1+\delta_1}ba^{t\gamma_2+\delta_2}.$$
Hence $g_1g_2=g_2g_1$ if and only if
$s\delta_1+\gamma_1=t\gamma_1+\delta_1$ and $s\delta_2+\gamma_2=t\gamma_2+\delta_2$.
This implies the claim.
\end{proof}

\section{Commuting morphisms $g_1$ and $g_2$ such that $g_1$ or $g_2$ is singular}
Let $X=\{a,b\}$ and let $h\in \Tri(X^*)$. If $h$ is singular, then $h(a)=\varepsilon$ or $h(b)\in a^*$.
In this section we give all commuting morphisms $g_i\in \Tri(X^*)$, $i=1,2$, such that $g_1$ or $g_2$ is singular.
\begin{proposition}
Let $g_i\in \Tri(X^*)$ for $i=1,2$. Assume that $g_1(b) \in a^*$. Then $g_1g_2=g_2g_1$ if and only if $|g_1g_2(b)|=|g_2g_1(b)|$  or, equivalently,
$$|g_1(a)||g_2(b)|_a+|g_1(b)||g_2(b)|_b=|g_2(a)||g_1(b)|.$$
\end{proposition}

\begin{proof}
Since $g_2g_1(b)\in a^*$ and $g_1g_2(b)\in a^*$, the claim holds.
\end{proof}

\begin{proposition}
Let $g_1,g_2\in \Tri(X^*)$ be morphisms such that
$$g_1(a)=\varepsilon, \hspace{3mm} g_1(b)=u$$
and
$$g_2(a)=a^t, \hspace{3mm} g_2(b)=v,$$
where $t\geq 0$ and both $u$ and $v$ have at least one occurrence of $b$. Then $g_1g_2=g_2g_1$ if and only if at least one of the following conditions holds: \smallskip

(i) $g_1=g_2$,

(ii) $g_2(x)=x$ for all $x\in X$,

(iii) $t=0$ and $uv=vu$,

(iv) $g_i(b)\in b^*$ for $i=1,2$,
\eject
(v) $t=1$ and there exist nonnegative integers $\alpha,\beta,i$ and $j$ such that
$$g_1(b)=(a^{\alpha}ba^{\beta})^i, \hspace{3mm} g_2(b)=(ba^{\alpha+\beta})^jb.$$
\end{proposition}

\begin{proof}
If (i), (ii), (iii), (iv) or (v) holds, then  $g_1g_2=g_2g_1$ (see Examples \ref{ex1}, \ref{ex6}, \ref{ex7}).

\medskip
Assume that $g_1g_2=g_2g_1$. Then
$$g_2(u)=g_2g_1(b)=g_1g_2(b)=g_1(v)=u^{|v|_b}.$$
If $u\in b^*$, this equation implies that $v\in b^*$. Hence (iv) holds. Assume that $u\not\in b^*$.

\medskip
Next, assume that $t=0$. Then $u^{|v|_b}=g_2(u)=v^{|u|_b}$, which shows that (iii) holds.

Assume then that $t\neq 0$.
By assumption, $|v|_b=1$ or $|v|_b\geq 2$. Assume first that $|v|_b=1$. Then the equation $g_2(u)=u$ shows that (ii) holds. Assume finally that $|v|_b\geq 2$.
Then $\omega(g_2)$ is defined. Since $\omega(g_2)$ is obtained from $g_2^{\omega}(u)=u^{\omega}$ by deleting the occurrences of $a$ preceding the first occurrence of $b$, the word
$\omega(g_2)$ is eventually periodic. Hence $t=1$. By Lemma \ref{le2} there are nonnegative integers $\gamma$ and $j$ such that
$g_2(b)=(ba^{\gamma})^jb$. Since $\omega(g_2)=(ba^{\gamma})^{\omega}$, there exist nonnegative integers $\alpha,\beta$ and $i$ such that
$g_1(b)=(a^{\alpha}ba^{\beta})^i$ and $\gamma=\alpha+\beta$. Hence (v) holds.
\end{proof}

For a systematic study of the equation $h(w)=w^n$, $n\geq 2$, for binary morphisms, see \cite{HL}.

\end{document}